\documentclass[12pt,twoside]{article}

 \usepackage{float}
\usepackage{graphicx}
\usepackage{epstopdf}
\usepackage{graphicx}
\usepackage{epic}
\usepackage{multirow}
\usepackage{pst-poly}  
\usepackage{pst-plot}  
\usepackage{pst-poly}  
\usepackage{tikz}
\usepackage{xcolor}
\usetikzlibrary{arrows,shapes,chains}

\renewcommand{\paragraph}{\roman{paragraph}}
\usepackage[a4paper]{geometry}
\makeatletter
\renewcommand\title[1]{\gdef\@title{\reset@font\Large\bfseries #1}}
\renewcommand\section{\@startsection {section}{1}{\z@}%
                                   {-3.5ex \@plus -1ex \@minus -.2ex}%
                                   {2.3ex \@plus.2ex}%
                                   {\normalfont\large\bfseries}}
\renewcommand\subsection{\@startsection{subsection}{2}{\z@}%
                                     {-3ex\@plus -1ex \@minus -.2ex}%
                                     {1.5ex \@plus .2ex}%
                                     {\normalfont\normalsize\bfseries}}
\renewcommand\subsubsection{\@startsection{subsubsection}{3}{\z@}%
                                     {-2.5ex\@plus -1ex \@minus -.2ex}%
                                     {1.5ex \@plus .2ex}%
                                     {\normalfont\normalsize\bfseries}}

\def\@runningauthor{}\newcommand{\runningauthor}[1]{\def\runningauthor{#1}}
\def\@runningtitle{}\newcommand{\runningtitle}[1]{\def\runningtitle{#1}}

\renewcommand{\ps@plain}{%
\renewcommand{\@evenhead}{\footnotesize\scshape \hfill\runningauthor\hfill}
\renewcommand{\@oddhead}{\footnotesize\scshape \hfill\runningtitle\hfill}}

\newcommand{\F}{\mathbb{F}}
\newcommand{\x}{\mathbf{x}}

\newcommand{\T}{{\rm T}}

\newcommand {\ccc}{{\mathbf{c}}}
\newcommand {\dd}{\mathbf{d}}
\newcommand {\0}{\mathbf{0}}
\newcommand {\aaa}{\alpha}
\newcommand {\bbb}{\beta}
\g@addto@macro\bfseries{\boldmath}

\makeatother

\usepackage{amsthm,amsmath,amssymb}
\usepackage{cite}
\usepackage{graphicx}

\usepackage[colorlinks=true,citecolor=black,linkcolor=black,urlcolor=blue]{hyperref}

\theoremstyle{plain}
\newtheorem{theorem}{Theorem}
\newtheorem{lemma}[theorem]{Lemma}
\newtheorem{corollary}[theorem]{Corollary}

\theoremstyle{definition}
\newtheorem{definition}[theorem]{Definition}
\newtheorem{example}[theorem]{Example}

\theoremstyle{remark}

\title{The $b$-symbol weight hierarchy of the Kasami codes
}

\runningtitle{The $b$-symbol weight hierarchy of the Kasami codes}

\author{Hongwei Zhu \thanks{ School of Mathematical Sciences, Anhui University, Hefei, China. E-mail: zhwgood66@163.com}
\and Minjia Shi\thanks{School of Mathematical Sciences, Anhui University, Hefei, China. E-mail: smjwcl.good@163.com}
}

\runningauthor{}

\date{}

\begin{document}

\maketitle

\thispagestyle{empty}

\begin{abstract}
The symbol-pair read channel was first proposed by Cassuto and Blaum. Later, Yaakobi et al. generalized it to the $b$-symbol read channel.  It is motivated by the limitations of the reading process in high density data storage systems. One main task in $b$-symbol coding theory is to determine the $b$-symbol weight hierarchy of codes. In this paper, we study the $b$-symbol weight hierarchy of the Kasami codes, which are well known for their applications to construct sequences with optimal correlation magnitudes. The complete symbol-pair weight distribution of the Kasami codes is determined.
\end{abstract}

\section{Introduction}
Let us introduce two metrics, which are two different generalizations of the Hamming metric. The first metric is the $b$-th generalized Hamming weight metric. It was proposed by Wei \cite{wei} to characterize the code's performance on the wire-tap channel of type II. The $b$-th generalized Hamming weight metric is defined as follows.
Let $C$ be an $[n,k]$ linear code over a finite field. For any subcode $D\subset C$, the support of $D$ is defined to be
$$\chi(D)=\{i:0\leq i\leq n-1|c_i\neq 0~\hbox{for some}~ (c_0,\ldots,c_{n-1})\in D\}.$$
The $b$-th generalized Hamming weight of a code $C$ is the smallest support of a $b$-dimensional subcode of $C$ with $1\leq b\leq k$.
We use $\dd_b(C)$ to denote the minimum $b$-th generalized Hamming distance of $C$.
 When $b=1$, $\dd_1(C)$ is the minimum Hamming distance of $C$.
 The set $$\{\dd_b(C)|1\leq b\leq k\}$$ is called the weight hierarchy of $C$. To distinguish it from the later definition, let us call it the generalized weight hierarchy in this paper.

The second metric is the $b$-symbol metric. Cassuto and Blaum \cite{CB,CB1} first proposed the concept of the symbol-pair channel in 2010. This new paradigm is motivated by the limitations of the reading process in high data storage systems.
 Later, Yaakobi et. al \cite{Yaa1} generalized the symbol-pair metric to $b$-symbol metric. The $b$-symbol metric is defined as follows.
Let $b$ be a positive integer with $1\leq b\leq n$. Let $F$ be a finite field and let $\x=(x_0,x_1,\ldots,x_{n-1})$ and $\mathbf{y}=(y_0,y_1,\ldots,y_{n-1})$ be two vectors which belong to $F^n$. The $b$-symbol weight of $\x$, denoted by $w_b(\x)$, is the Hamming weight of $\pi_b(\x)$, where $\pi_b(\x)\in(F^b)^n$ and
      $$\pi_b(\x)=((x_0,\ldots,x_{b-1}),(x_1,\ldots,x_b),\ldots,
      (x_{n-1},\ldots,x_{b+n-2({\rm mod~}n)})).$$
      $w_1(\x)$ denotes the Hamming weight of $\x$.
       The $b$-symbol distance of $\x$ and $\mathbf{y}$, denoted by $d_b(\x,\mathbf{y})$, is the $b$-symbol weight of $\x-\mathbf{y}$. For a code $C$, $d_b(C)$ denotes the minimum $b$-symbol distance of $C$. When $b=1$, $d_1(C)$ is also the minimum Hamming distance of $C$.
      The $b$-symbol metric is also called symbol-pair metric if $b=2$.
       The set
      $$\{d_b(C)|1\leq b\leq n\}$$
      is called the $b$-symbol weight hierarchy of $C$.
Note that $C$ could be an unrestricted code under the $b$-symbol metric. If $C$ is a cyclic code (or a constacyclic code), then the $b$-symbol weight hierarchy of $C$ has the following property:
$$d_1(C)<d_2(C)<\cdots<d_{k-1}(C)<d_k(C)=d_{k+1}(C)=\cdots=d_n(C)=n.$$
The generalized Hamming weight hierarchy of $C$ has a similar property, and $C$ could not be cyclic. For the same cyclic code $C$, it's minimum $b$-th generalized Hamming distance and minimum $b$-symbol distance have the following relationship.
\begin{theorem}{\rm\cite{BUG}}\label{thm1}
If $C$ is cyclic, then $d_b(C)\geq \dd_b(C)$ with $1\leq b\leq k.$
\end{theorem}
For more details about the connections between these two metrics we refer the reader to \cite{BUG}. Since we have explained what $b$-symbol weight hierarchy is, we elaborate on our motivation for studying the $b$-symbol weight hierarchy of the Kasami codes.

Recently, many scholars have paid close attention to the $b$-symbol metric and extensively studied various properties \cite{CL,C+,C+1,CLL,DGZ,DTG,Eli,KZL,Yaa}.
One main task in $b$-symbol coding theory is to determine the $b$-symbol weight hierarchy of codes.
 As we all know, it is very difficult to determine the Hamming weight distribution or the generalized weight hierarchy of cyclic codes. The problems to determine the $b$-symbol weight hierarchy and to complete the $b$-symbol weight distribution are likely to be more complicated than the preceding two problems. So far, many contributions to the Hamming weight distributions and generalized Hamming hierarchy of cyclic codes concentrate on irreducible cyclic codes or the duals of cyclic codes with two zeroes \cite{DTWD,DingLiuMaZeng,DY,FengMomi,Machangli,WangTangQiYang,Xiong1FFA,
Xiong2DCC,Xiong3FFA}.

Our motivation in this paper is to investigate the $b$-symbol weight hierarchy of a class of cyclic codes.
To the best of our knowledge, the contributions on the $b$-symbol weight distribution or $b$-symbol weight hierarchy are the following.
\begin{itemize}
  \item Shi et al. \cite{SOS,ZHW} studied the $b$-symbol weight distribution of two classes of irreducible cyclic codes.
  \item Sun et al. \cite{Z} studied the symbol-pair distance distribution of a class of repeated-root cyclic codes.
\end{itemize}

 A natural idea is to study a class of the duals of cyclic codes with two zeroes. Then the Kasami codes are nice research objects, which are introduced by Kasami in 1966 \cite{Kasami}.
 These codes are of importance for many applications since they are used to construct sequences with optimal correlation magnitudes.
  They are a class of the duals of cyclic codes with two zeroes $\theta, \theta^{2^m+1}$, where $\theta$ denotes a primitive element of $\F_{2^{2m}}$. The Hamming weight distribution and the generalized weight hierarchy of the Kasami codes were given by Helleseth and Kumar in \cite{TorKasami}.
  Recently, Shi et al. \cite{SK} studied a new derivation of the Hamming weight distribution and the coset graph of the Kasami codes.
  Here we list their generalized weight hierarchies since they are lower bounds of the minimum $b$-symbol distances of the Kasami codes.
\begin{theorem}\label{thm0}{\rm \cite{TorKasami}}
The generalized Hamming weight hierarchy of the $[2^{2m}-1,3m,2^{2m-1}-2^{m-1}]$ Kasami codes is given by
\begin{equation*}
  \dd_b(C)=\left\{
             \begin{array}{ll}
               (2^b-1)(2^{2m-b}-2^{m-b}), & \hbox{if $1\leq b\leq m$;} \\
               2^{2m}-2^m-2^{2m-b}+1, & \hbox{if $m<b\leq 2m$;} \\
               2^{2m}-2^{3m-b}, & \hbox{if $2m<b\leq 3m$.}
             \end{array}
           \right.
\end{equation*}
\end{theorem}

In this paper, we study the $b$-symbol weight hierarchy of the Kasami codes. Just as the research process of the generalized weight hierarchy of the Kasami codes, we classify and discuss as follows: (I) $1\leq b\leq m$; (II) $m<b\leq 2m$; (III) $2m<b\leq 3m$; (IV) $3m+1\leq b\leq n$. We obtain our results up to a new invariant $\#\mathcal{T}_3(j;\aaa,\bbb)$ that we introduce in subsection 3.3. When $b=2$, we find the complete symbol-pair weight distribution of the Kasami codes. Under some special conditions, the shortening of the Kasami codes $C_{\overline{\mathcal{I}_b(\ccc)}}$ are the Griesmer codes when $1\leq b\leq 2m,$ where $\ccc$ is a codeword with the minimum nonzero $b$-symbol weight. The definition of $\overline{\mathcal{I}_b(\ccc)}$ appears in Definition \ref{support}.

The rest of this paper is organized as follows. In Section 2, we introduce some basic notations, definitions and a useful lemma. In Section 3, we show the main results and related proofs. In the Conclusion, we conclude this paper.
\section{Preliminaries}
Let $m$ be a positive integer and $q=2^m$. Let $\F_q$ denote the finite field with $q$ elements and $\F_q^*=\F_q\setminus\{0\}$. Let the Frobenius trace from $\F_{q^k}$ to $\F_q$ be defined by
$${\rm Tr}_{\F_{q^k}/{\F_q}}(x)=\sum_{i=0}^{k-1}x^{q^i}.$$
For the sake of convenience, we denote the trace functions from $\F_{2^k}$ to $\F_2$ by $\T_k.$
\begin{definition}
Let $\ccc(\aaa,\bbb)$ with $(\aaa,\bbb)\in\F_{q^2}\times\F_q$ be a vector of length $2^{2m}-1$, indexed by the elements of $\F_{q^2}^*$,
\begin{equation*}
 \ccc(\aaa,\bbb)=(\T_{2m}(\aaa x)+\T_m(\bbb x^{q+1}), x\in \F_{q^2}^*).
\end{equation*}
Then the Kasami code is a linear code with parameters $[2^{2m}-1,3m,2^{2m-1}-2^{m-1}]$ defined by
$$C=\{\ccc(\aaa,\bbb)|\aaa\in\F_{q^2},\bbb\in\F_q\}.$$
Let $\theta$ be a primitive element of $\F_{q^2}.$ Let $x$ run through
$\F_{q^2}^*$ in the following order
\begin{equation}\label{K1}
  \theta,\theta^2,\theta^3,\ldots,\theta^{q^2-1}.
\end{equation}
Let $\eta=\theta^{q+1}$ which is a primitive element of $\F_q.$
Therefore, $\ccc(\aaa,\bbb)$ can be written as
$$(\T_{2m}(\aaa \theta)+\T_m(\bbb \eta),\T_{2m}(\aaa \theta^2)+\T_m(\bbb \eta^2),\ldots,\T_{2m}(\aaa \theta^{q^2-1})+\T_m(\bbb \eta^{q^2-1})).$$
\end{definition}
For a vector $\x=(x_0,x_1,\ldots,x_{n-1})$, let $\tau(x)$ denote the vector $(x_{n-1},x_0,\ldots,x_{n-2})$ obtained from $\x$ by the cyclic shift of the coordinates $i\mapsto i+1 {\rm~ mod~}n.$

The following result plays an indispensable role in the computation of Hamming weight distribution and generalized weight hierarchy of the Kasami codes. Of course, it is still very important to this paper.
\begin{lemma}{\rm \cite{TorKasami}}\label{lem2}
Let $S(\aaa,\bbb)$ be defined by
 $$S(\aaa,\bbb)=\sum_{x\in\F_{q^2}^*}(-1)^{\T_{2m}(\aaa x)+\T_m(\bbb x^{q+1})}.$$
Then
\begin{equation*}
  S(\aaa,\bbb)=\left\{
                 \begin{array}{ll}
                   q^2-1, & \hbox{if $\bbb=0$ and $\aaa=0$;} \\
                   -1, & \hbox{if $\bbb=0$ and $\aaa\neq 0$;} \\
                   -q-1, & \hbox{if $\bbb\neq 0$ and $\T_m(\frac{\aaa^{q+1}}{\bbb})=0$;} \\
                   q-1, & \hbox{if $\bbb\neq 0$ and $\T_m(\frac{\aaa^{q+1}}{\bbb})=1$.}
                 \end{array}
               \right.
\end{equation*}
\end{lemma}
\section{The $b$-symbol weight hierarchy}
In this section, we will find the $b$-symbol weight hierarchy of the Kasami codes. Let $F$ be a finite field. For any vector $\ccc \in F^n$,
the following lemma gives the relation between $w_b(\ccc)$ and $w_1(\ccc).$
\begin{lemma}{\rm\cite{BUG}}\label{lem3}
Let $\mathbf{c}\in F^n$ and denote by $V_b(\mathbf{c})$ the codewords generated by all linear combinations of $\ccc$ and its first $b-1$ cyclic shifts. Then
$$w_b(\mathbf{c})=\frac{1}{|F|^{b-1}(|F|-1)}\sum_{\mathbf{c}^{\prime}\in V_b(\mathbf{c})}w_1(\mathbf{c}^{\prime}).$$
\end{lemma}
\begin{theorem}\label{thm4}
Let $\ccc(\aaa,\bbb)$ be a codeword of the Kasami code. Then
\begin{eqnarray}\label{K2}
  w_b(\ccc(\aaa,\bbb)) &=& \frac{1}{2^{b-1}}\sum_{(u_1,u_2,\ldots,u_b)\in \F_2^b}
w_1(\ccc(\aaa(\sum_{i=1}^bu_i\theta^{i-1}),\bbb(\sum_{i=1}^b
u_i\eta^{i-1}))),
\end{eqnarray}
and\begin{small}
\begin{eqnarray}\label{33}
  \nonumber   w_b(\ccc(\aaa,\bbb)) &=&\frac{1}{2^{b}}{\Bigg{[}}(2^b-1)(q^2-1)
-b\cdot S(\aaa,\bbb)-\\
~&~&\sum_{j=1}^{b-1}(b-j)
\sum_{(u_1,\ldots,u_{j-1})\in\F_2^{j-1}}
S(
\aaa(1+\sum_{i=1}^{j-1}u_i\theta^{i}+\theta^j),
\bbb(1+\sum_{i=1}^{j-1}u_i\eta^{i}+\eta^j)
){\Bigg{]}}.
\end{eqnarray}
\end{small}
\end{theorem}
\begin{proof}
Let $x$ run through $\F_{q^2}^*$ in the order given in (\ref{K1}). Then $\ccc(\aaa,\bbb)$ and the first $b-1$ cyclic shifts of $\ccc(\aaa,\bbb)$ are
\begin{eqnarray*}
  \ccc(\aaa,\bbb) &=&  (\T_{2m}(\aaa \theta)+\T_m(\bbb \eta),\ldots,\T_{2m}(\aaa \theta^{q^2-1})+\T_m(\bbb \eta^{q^2-1})),\\
  \tau(\ccc(\aaa,\bbb)) &=&  (\T_{2m}(\aaa \theta^2)+\T_m(\bbb \eta^{2}),\ldots,\T_{2m}(\aaa \theta)+\T_m(\bbb \eta))=\ccc(\aaa\theta,\bbb\eta),\\
  ~ &\vdots&~  \\
  \tau^{b-1}(\ccc(\aaa,\bbb)) &=& (\T_{2m}(\aaa \theta^{b})+\T_m(\bbb \eta^{b}),\ldots,\T_{2m}(\aaa \theta^{b-1})+\T_m(\bbb \eta^{b-1}))
 = \ccc(\aaa\theta^{b-1},\bbb\eta^{b-1}).
\end{eqnarray*}
According to the definition of $V_b(\ccc(\aaa,\bbb))$,
\begin{eqnarray*}
  V_b(\ccc(\aaa,\bbb)) &=& \left\{\left.\sum_{i=1}^{b}u_i\ccc(\aaa\theta^{i-1},\bbb\eta^{i-1})\right|
(u_1,\ldots,u_b)\in\F_2^b\right\} \\
  ~ &=&\left\{\left.\ccc\left(\aaa\sum_{i=1}^bu_i\theta^{i-1}
  ,\bbb\sum_{i=1}^b
u_i\eta^{i-1}\right)\right|
(u_1,\ldots,u_b)\in\F_2^b\right\}.
\end{eqnarray*}
Eq.(\ref{K2}) follows from Lemma \ref{lem3}.
There is a nice relationship between the $w_1(\ccc(\aaa,\bbb))$ and $S(\aaa,\bbb)$, that is,
$$w_1(\ccc(\aaa,\bbb))=\frac{1}{2}\left(q^2-1-S(\aaa,\bbb)\right).$$
In the light of Eq.(\ref{K2}), we have
\begin{equation}\label{K3}
  w_b(\ccc(\aaa,\bbb))=\frac{1}{2^b}\left[(2^b-1)(q^2-1)-
\sum_{(u_1,\ldots,u_b)\in\F_2^b\setminus\{\0\}}S\left(\aaa\sum_
{i=1}^bu_i\theta^{i-1}
,\bbb\sum_{i=1}^bu_i\eta^{i-1}\right)\right].
\end{equation}
Since $\T_m\left(\frac{\aaa^{q+1}}{\bbb}\right)=\T_m\left(\frac{(\aaa\cdot \theta)^{q+1}}{\bbb\eta}\right)$, by Lemma \ref{lem2}, $S(\aaa,\bbb)=S(\aaa\theta,\bbb\eta).$
Eq.(\ref{33}) follows by merging these identical items.
This completes the proof.
\end{proof}

We adopt the convention that the term $\sum_{i=1}^{j-1}u_i\theta^{i}$ vanishes if $j=1$. For convenience, let $\Theta_j$ and $\Omega_j$ be defined by
\begin{equation*}
  \Theta_j = 1+\sum_{i=1}^{j-1}u_i\theta^i+\theta^j \hbox{~~and~~ $\Omega_j = 1+\sum_{i=1}^{j-1}u_i\eta^i+\eta^j,$}
\end{equation*}
where $\Theta_1=1+\theta$ and $\Omega_1=1+\eta.$
\subsection{Case $1\leq b\leq m$}
In this subsection, we assume that $1\leq b\leq m$.
For any $1\leq j\leq b-1$ and $(\aaa,\bbb)\in\F_{q^2}^*\times\F_q^*$, let $\mathcal{T}_1(j;\aaa,\bbb)$ be defined by
$$\mathcal{T}_1(j;\aaa,\bbb)=
\left\{(u_1,u_2,\ldots,
u_{j-1})\in\F_2^{j-1}\left|\T_m\left(\frac{\aaa^{q+1}}{\bbb}\cdot
\frac{\Theta_j
^{q+1}}{\Omega_j}\right.\right)=1\right\}.$$
Let $\#\mathcal{T}_1(j;\aaa,\bbb)$ denote the size of $\mathcal{T}_1(j;\aaa,\bbb)$. A trivial upper bound of $\#\mathcal{T}_1(j;\aaa,\bbb)$ is $2^{j-1}.$
Note that the denominator $\Omega_j$ can not be zero since $1\leq b\leq m$.
\begin{lemma}\label{lem7}
Let $b$ be a positive integer. Then
$$\sum_{j=1}^{b-1}(b-j)2^{j-1}=2^b-1-b.$$
\end{lemma}
\begin{proof}
The proof is straightforward and omitted.
\end{proof}
\begin{theorem}\label{cor1}
 Let $(\aaa,\bbb)\in\F_{q^2}\times\F_q$ and $1\leq b\leq m$. Then
\begin{equation*}
  w_b(\ccc(\aaa,\bbb))=
\left\{
  \begin{array}{ll}
    0, & \hbox{if $\aaa=0$ and $\bbb=0$;} \\
    (2^b-1)2^{2m-b}, & \hbox{if $\aaa\neq0$ and $\bbb=0$;}\\
    (2^b-1)(2^{2m-b}+2^{m-b}), &\hbox{if $\aaa=0$ and $\bbb\neq0$.}
  \end{array}
\right.
\end{equation*}
For $(\aaa,\bbb)\in\F_{q^2}^*\times\F_q^*$, we have
\begin{eqnarray*}
 w_b(\ccc(\aaa,\bbb)) &=&
2^{2m}+2^{m}-2^{2m-b}-2^{m-b}-\frac{b(S(\aaa,\bbb)+2^m+1)}{2^b}  \\
  ~ &~&
-2^{m+1-b}
\sum_{j=1}^{b-1}(b-j)\#\mathcal{T}_1(j;\aaa,\bbb).
\end{eqnarray*}
\end{theorem}
\begin{proof}
If $1\leq b\leq m$, then
$\Theta_j
$ and $\Omega_j$ can not be zero.
From Eq.(\ref{K3}) and Lemma \ref{lem2}, we have
\begin{equation*}
  w_b(\ccc(\aaa,\bbb))=\left\{
                         \begin{array}{ll}
                           0, & \hbox{if $\aaa=\bbb=0$;} \\
                           \frac{1}{2^b}((2^b-1)(q^2-1)-
(2^b-1)\cdot(-1)), & \hbox{if $\aaa\neq0$ and $\bbb=0$;} \\
     \frac{1}{2^b}((2^b-1)(q^2-1)-
(2^b-1)\cdot(-q-1)), & \hbox{if $\aaa=0$ and $\beta\neq0$.}
                         \end{array}
                       \right.
\end{equation*}
Let $(\aaa,\bbb)\in\F_{q^2}^*\times\F_q^*$.
Combining Lemma \ref{lem2} and Theorem \ref{thm4}, we have
\begin{eqnarray*}\label{KK1}
  \nonumber w_b(\ccc(\aaa,\bbb)) &=&\frac{1}{2^b}\left[(2^b-1)(q^2-1)-b S(\aaa,\bbb)-\sum_{j=1}^{b-1}(b-j)\sum_{(u_1,\ldots,u_{j-1})
  \in\F_2^{j-1}}S(\aaa\Theta_j,\bbb\Omega_j)\right]\\
  \nonumber ~&=& \frac{1}{2^b}{\Bigg{[}}(2^b-1)(q^2-1)-b
  S(\aaa,\bbb)-\sum_{j=1}^{b-1}(b-j)
  \left(\#\mathcal{T}_1(j;\aaa,\bbb)(q-1)\right.\\
  \nonumber  ~&~&\left.+(2^{j-1}-\#\mathcal{T}_1(j;\aaa,\bbb))(-q-1)\right)
 {\Bigg{]}}\\
 \nonumber ~ &=& \frac{1}{2^b}\left[(2^b-1)(q^2-1)-b
  S(\aaa,\bbb)-2^{m+1}\sum_{j=1}^{b-1}(b-j)\#\mathcal{T}_1(j;\aaa,\bbb)\right.\\
  \nonumber ~&~&\left.+(2^m+1)\sum_{j=1}^{b-1}(b-j)2^{j-1}\right]\\
  \nonumber ~&=&2^{2m}+2^{m}-2^{2m-b}-2^{m-b}-\frac{b( S(\aaa,\bbb)+2^m+1)}{2^b}  \\
  ~ &~&
-2^{m+1-b}
\nonumber\sum_{j=1}^{b-1}(b-j)\#\mathcal{T}_1(j;\aaa,\bbb).{\hbox{~~~~(by Lemma \ref{lem7})}}
\end{eqnarray*}
This completes the proof.
\end{proof}
An important property of trace functions is that $\T_m(\beta)$ takes on each value in $\F_2$ equally often, i.e., $2^{m-1}$ times, where $\beta\in\F_q$. The following lemma is a generalization of the preceding property of trace functions and a well known result follows from the properties of $m$-sequences. For completeness, we give a proof.
\begin{lemma}\label{lem6}
Let $\{\beta_1,\beta_2,\ldots,\bbb_m\}$ be a basis of $\F_q$ over $\F_2$. Let
$$U_i=\{x\in\F_q|\T_m(x\beta_i)=1\}$$ 
 with $1\leq i\leq s$. Then
$$\left|\cap_{i\in\mathcal{M}}U_i\right|=2^{m-|\mathcal{M}|},$$
where $\mathcal{M}$ is a subset of $\{1,2,\ldots,m\}.$
\end{lemma}
\begin{proof}
Assume that $\mathcal{M}=\{i_1,i_2,\ldots,i_t\}.$
The element $x$ in $\cap_{i\in\mathcal{M}}U_i$ is a solution of the system of equations
\begin{equation*}
  \left\{
     \begin{array}{ll}
       \T_m(x\beta_{i_1})=1; \\
       \T_m(x\beta_{i_2})=1; \\
       ~~~~~~~~\vdots \\
       \T_m(x\beta_{i_t})=1.
     \end{array}
   \right.
\end{equation*}
By \cite[Lemma 3.51]{finite}, the determinant $\det\left((\bbb_{i}^{q^{j-1}})_{m\times m}\right)$ is not zero if and only if $\{\beta_1,\beta_2,\ldots,\bbb_m\}$ are linearly independent over $\F_2$. Then the desired result follows.
\end{proof}
For any $1\leq j\leq b-1$,
let $\mathcal{A}_1(j)$ be defined by
$$\mathcal{A}_1(j)=\left\{\left.\frac{\Theta_j^{q+1}}{\Omega_j}
\right|(u_1,u_2,\ldots,
u_{j-1})\in\F_2^{j-1}\right\}.$$
A trivial upper bound of the size of $\mathcal{A}_1(j)$ is $2^{j-1}$.
\subsubsection{Case $b=2$}
If $b=2$, it is easy to check that $1$ and $\frac{(1+\theta)^{q+1}}{1+\eta}$ are linearly independent over $\F_2$.
The complete symbol-pair weight distribution of the Kasami code is determined.
\begin{theorem}\label{symbolpair}
For $m\geq 2$, the symbol-pair weight enumerator of the Kasami code is\begin{eqnarray*}
    A(T) &=& 1+(2^{3m-2}-2^{m-2})(T^{3\cdot(2^{2m-2}-2^{m-2})}+
T^{3\cdot2^{2m-2}-2^{m-2}}+
T^{3\cdot2^{2m-2}+2^{m-2}})+ \\
    ~ &~& (2^{3m-2}-2^{2m}+3\cdot2^{m-2})
T^{3(2^{2m-2}+2^{m-2})}+(2^{2m}-1)T^{3\cdot2^{2m-2}}.
  \end{eqnarray*}
Moreover,
\begin{equation*}
  w_2(\ccc(\aaa,\bbb))=
  \left\{
    \begin{array}{ll}
      0, & \hbox{if $\bbb=\aaa=0$ $($once$)$;} \\
      3\cdot2^{2m-2}, & \hbox{if $\bbb=0$ and $\aaa\neq0$ $($$2^{2m}-1$ times$)$;} \\
      3\cdot(2^{2m-2}+2^{m-2}), & \hbox{if $\bbb\neq0$ and $\T_m\left(\frac{\aaa^{q+1}}{\bbb}\right)=
       \T_m\left(\frac{\aaa^{q+1}}{\bbb}\cdot\frac{(1+\theta)^{q+1}}
       {1+\eta}\right)=0$} \\
       ~&\hbox{$($$2^{3m-2}-2^{2m}+3\cdot2^{m-2}$ times$)$;} \\
      3\cdot(2^{2m-2}-2^{m-2}), & \hbox{if $\bbb\neq0$ and $\T_m\left(\frac{\aaa^{q+1}}{\bbb}\right)=
       \T_m\left(\frac{\aaa^{q+1}}{\bbb}\cdot\frac{(1+\theta)^{q+1}}
       {1+\eta}\right)=1$}\\
       ~&\hbox{$($$2^{3m-2}-2^{m-2}$ times$)$;} \\
      3\cdot2^{2m-2}-2^{m-2}, & \hbox{if $\bbb\neq0$ and $\T_m\left(\frac{\aaa^{q+1}}{\bbb}\right)=1$ and }\\
       ~&\hbox{$
       \T_m\left(\frac{\aaa^{q+1}}{\bbb}\cdot\frac{(1+\theta)^{q+1}}
       {1+\eta}\right)=0$$~($$2^{3m-2}-2^{m-2}$ times$)$;} \\
      3\cdot2^{2m-2}+2^{m-2}, & \hbox{if $\bbb\neq0$ and $\T_m\left(\frac{\aaa^{q+1}}{\bbb}\right)=0$ and }\\
       ~&\hbox{$
       \T_m\left(\frac{\aaa^{q+1}}{\bbb}\cdot\frac{(1+\theta)^{q+1}}
       {1+\eta}\right)=1$ $($$2^{3m-2}-2^{m-2}$ times$)$.}
    \end{array}
  \right.
\end{equation*}
\end{theorem}
\begin{proof}
The result is clearly true if $\aaa=\bbb=0$. It follows from Lemma \ref{lem2} that the case $\bbb=0$ and $\aaa\neq0$ holds. In the remaining case, we assume that $\bbb\neq 0.$ Since $\frac{(1+\theta)^{q+1}}
       {1+\eta}\neq 1$, $\frac{\aaa^{q+1}}{\bbb}$ and $\frac{\aaa^{q+1}}{\bbb}\cdot\frac{(1+\theta)^{q+1}}
       {1+\eta}$ are linearly independent over $\F_2$ if $\aaa\neq0$.
\begin{itemize}
  \item (i) If
\begin{equation}\label{K8}
\T_m\left(\frac{\aaa^{q+1}}{\bbb}\right)=
       \T_m\left(\frac{\aaa^{q+1}}{\bbb}\cdot\frac{(1+\theta)^{q+1}}
       {1+\eta}\right)=1,
\end{equation}
then $\aaa$ can not be zero. By Lemma \ref{lem6}, there are $(q^2-1)2^{m-2}=2^{3m-2}-2^{m-2}$ $(\aaa,\bbb)\in\F_{q^2}^*\times\F_q^*$ such that
Eq.(\ref{K8}) holds.
Then
\begin{eqnarray*}
  w_b(\ccc(\aaa,\bbb)) &=& \frac{1}{4}\left(2^{2+2m}-2^{2m}+2^{2+m}-3\cdot2^m-2-2^{m+1}+2-2^{m+1}
  \right) \\
  ~ &=& 3\left(2^{2m-2}-2^{m-2}\right).
\end{eqnarray*}
  \item (ii) If
\begin{equation}\label{K9}
\T_m\left(\frac{\aaa^{q+1}}{\bbb}\right)=1~\hbox{and}~
       \T_m\left(\frac{\aaa^{q+1}}{\bbb}\cdot\frac{(1+\theta)^{q+1}}
       {1+\eta}\right)=0,
\end{equation}
then $\aaa$ can not be zero. Let $\overline{U_i}$ denote the complementary set of $U_i$. Since $|U_1|=|U_2|=2^{m-1}$ and $|U_1\cap U_2|=2^{m-2}$, $|U_1\cap \overline{U_2}|=2^{m-1}-2^{m-2}=2^{m-2}.$ Then there are $(q^2-1)2^{m-2}$ $(\aaa,\bbb)\in\F_{q^2}^*\times\F_q^*$ such that
Eq.(\ref{K9}) holds. Then
\begin{eqnarray*}
  w_b(\ccc(\aaa,\bbb)) &=& \frac{1}{4}\left(2^{2+2m}-2^{2m}+2^{2+m}-3\cdot2^m-2-
  2^{m+1}+2-2^{m+1}\cdot0\right) \\
  ~ &=& 3\cdot2^{2m-2}-2^{m-2}.
\end{eqnarray*}
  \item (iii) The frequency of the case that \begin{equation}\label{K10}
\nonumber\T_m\left(\frac{\aaa^{q+1}}{\bbb}\right)=0~\hbox{and}~
       \T_m\left(\frac{\aaa^{q+1}}{\bbb}\cdot\frac{(1+\theta)^{q+1}}
       {1+\eta}\right)=1
\end{equation}
is the same as the case (ii). The value of $w_b(\ccc(\aaa,\bbb))$ is
\begin{eqnarray*}
  w_b(\ccc(\aaa,\bbb)) &=& \frac{1}{4}\left(2^{2+2m}-2^{2m}+2^{2+m}-3\cdot2^m-2
  +2^{m+1}+2-2^{m+1}\right) \\
  ~ &=& 3\cdot2^{2m-2}+2^{m-2}.
\end{eqnarray*}
  \item (iv) If \begin{equation}\label{K11}
\T_m\left(\frac{\aaa^{q+1}}{\bbb}\right)=
       \T_m\left(\frac{\aaa^{q+1}}{\bbb}\cdot\frac{(1+\theta)^{q+1}}
       {1+\eta}\right)=0,
\end{equation}
then there are $$q^2\cdot(q-1)-3(q^2-1)2^{m-2}=2^{3m-2}-2^{2m}+3\cdot2^{m-2}$$
$(\aaa,\bbb)\in\F_{q^2}\times\F_q^*$ such that Eq.(\ref{K11}) holds. Then
\begin{eqnarray*}
  w_b(\ccc(\aaa,\bbb)) &=& \frac{1}{4}\left(2^{2+2m}-2^{2m}+2^{2+m}
  -3\cdot2^m-2+2^{m+1}+2-2^{m+1}\cdot0\right) \\
  ~ &=& 3\left(2^{2m-2}+2^{m-2}\right).
\end{eqnarray*}
\end{itemize} This completes the proof.
\end{proof}
\begin{example}\label{ex13}
(i) When $m=2$ and $b=2$, the Kasami code has parameters $[n=15,k=6,d_1(C)=6]$. Its symbol-pair weight distribution is
$1+15T^9+15T^{11}+15T^{12}+15T^{13}+3T^{15},$
and the minimum symbol-pair weight is $d_2(C)=9.$

(ii) When $m=3$ and $b=2$, the Kasami code has parameters $[n=63,k=9,d_1(C)=28]$. Its symbol-pair weight distribution is
$1+126T^{42}+126T^{46}+63T^{48}+126T^{50}+70T^{54},$
and the minimum symbol-pair weight is $d_2(C)=42.$

(iii) When $m=4$ and $b=2$, the Kasami code has parameters $[n=255,k=12,d_1(C)=120]$. Its symbol-pair weight distribution is
$1+1020T^{180}+1020T^{188}+255T^{192}+1020T^{196}+780T^{204},$
and the minimum symbol-pair weight is $d_2(C)=180.$

These results are verified by Magma programs.
\end{example}
Let $m(b)$ be the maximum integer with $2\leq m(b)\leq b-1$ such that the size of the set
$\{1\}\cup\bigcup_{i=1}^{m(b)}\mathcal{A}_1(i)$
is $2^{m(b)}$, and all elements which belong to
$\{1\}\cup\bigcup_{i=1}^{m(b)}\mathcal{A}_1(i)$ are linearly independent over $\F_2$.

The following lemma is useful in the proof of the subsequent theorem.
\begin{lemma}\label{lem11}
The following equation holds
$$\sum_{j=1}^{m(b)}(b-j)2^{j-1}=(b-m(b)+1)2^{m(b)}-b-1.$$
\end{lemma}
\begin{proof}
The proof is straightforward and omitted.
\end{proof}
\begin{theorem}
If $1\leq b\leq m$, then the minimum $b$-symbol distance of the Kasami code is no less than
$$
2^{2m}-2^{2m-b}+2^{m-b}+2^m(1-(b-m(b)+1)2^{1+m(b)-b}).$$
Moreover, if $m(b)=b-1$, then $d_b(C)=\dd_b(C)=(2^b-1)(2^{2m-b}-2^{m-b}).$
\end{theorem}
\begin{proof}
Combining the definition of $m(b)$ and Lemma \ref{lem6}, there exist $2^{m-2^{m(b)}}$ $x$ in $\F_q^*$ such that
$T_m(xy)=1$ for any $y\in \{1\}\cup\bigcup_{i=1}^{m(b)}\mathcal{A}_1(i).$
Then there exists $2^{m-2^{m(b)}}$ $(\aaa,\bbb)\in \F_{q^{2}}^*\times\F_q^*$ such that
$\#\mathcal{T}_1(j;\aaa,\bbb)=2^{j-1} \hbox{~and~} S(\aaa,\bbb)=q-1$
for any $1\leq j\leq m(b).$
According to Theorem \ref{cor1}, we have
\begin{eqnarray*}
  w_b(\ccc(\aaa,\bbb)) &=& \frac{1}{2^b}{\Bigg{[}}2^{b+2m}+2^{b+m}-2^{2m}-(b+1)2^m-b-b(2^m-1)
 \\
  ~ &~& -2^{m+1}\sum_{j=1}^{m(b)}(b-j)2^{j-1}-2^{m+1}\sum_{j=m(b)+1}^{b-1}
  (b-j)\#\mathcal{T}_1(j;\aaa,\bbb){\Bigg{]}} \\
  ~ &\geq&\frac{1}{2^b}\left[2^{b+2m}+2^{b+m}-2^{2m}-(2b+1)2^m
-2^{m+1}\sum_{j=1}^{m(b)}(b-j)2^{j-1}\right]\\
~&=&2^{2m}-2^{2m-b}+2^{m-b}+2^m(1-(b-m(b)+1)2^{1+m(b)-b}).{\hbox{~~(by Lemma \ref{lem11})}}
\end{eqnarray*}
Therefore, the first statement holds.

 If $m(b)=b-1$, then
 $w_b(\ccc(\aaa,\bbb))= 2^{2m}-2^m-2^{2m-b}+2^{m-b} =(2^b-1)(2^{2m-b}-2^{m-b})=\dd_b(C).$
Since $\dd_b(C)$ is a lower bound for $d_b(C)$, the desired result follows.
\end{proof}
\begin{example}
Notice that $m(b)=b-1$ only if $m\geq 2^{b-1}.$
The following numerical examples were computed by Magma programs. To our surprise we observed that if $m=2^{b-1}$, then $m(b)=b-1$ where $b=3,4,5,6$. It is reasonable to conjecture that $m(b)=b-1$ if $m= 2^{b-1}.$ We need more numerical examples to support this conjecture, but the example for $m$ greater than $6$ is too large to compute.
\begin{center}
\begin{tabular}{c|c}
  \hline
  $(m,b)$ & $m(b)$ \\
  \hline
  $(3,3)$ & $1$ \\
  $(4,3)$ & $2=b-1$ \\
  $(i,4)$ with $4\leq i\leq 7$ & $2$ \\
  $(8,4)$ & $3=b-1$ \\
  $(i,5)$ with $5\leq i\leq 7$ & $2$ \\
  $(i,5)$ with $8\leq i\leq 15$ & $3$ \\
  $(16,5)$ & $4=b-1$ \\
  $(i,6)$ with $6\leq i\leq 7$ & $2$ \\
  $(i,6)$ with $8\leq i\leq 15$ & $3$ \\
  $(i,6)$ with $16\leq i\leq 31$ & $4$ \\
  $(32,6)$ & $5=b-1$ \\
  \hline
\end{tabular}
\end{center}
\end{example}
\begin{definition}\label{support}
The $b$-symbol support of a vector $\x$ is defined by
$$\mathcal{I}_b(\x)=supp(\pi_b(\x))=\bigcup_{i=0}^{b-1}supp(\tau^i(\x)),$$
where $supp(\x)$ denotes the support of the vector $\x.$
Let $\overline{\mathcal{I}_b(\x)}=\{1,2,\ldots,n\}\setminus\mathcal{I}_b(\x).$
\end{definition}
\begin{corollary}\label{cor13}
Let $\ccc_0$ be a codeword with $w_b(\ccc_0)=d_b(C)$ of the Kasami code.
If $m(b)=b-1$, then the shortened code $C_{\overline{\mathcal{I}_b(\ccc_0)}}$ is a Griesmer code.
\end{corollary}
\begin{proof}
According to Theorem 11, the minimum $b$-symbol distance of the Kasami code is $(2^b-1)(2^{2m-b}-2^{m-b}).$ From \cite[Lemma 16]{BUG}, the matrix
\begin{equation*}
  G_b(\ccc_0)=\left(
                \begin{array}{c}
                  \ccc_0 \\
                  \tau(\ccc_0) \\
                  \vdots \\
                  \tau^{b-1}(\ccc_0) \\
                \end{array}
              \right)_{b\times n}
\end{equation*} has rank $b$.
Then the shortened code $C_{\overline{\mathcal{I}_b(\ccc_0)}}$ has parameters $$[n=d_b(C)=(2^b-1)(2^{2m-b}-2^{m-b}),k=b,
d_1(C_{\overline{\mathcal{I}_b(\ccc_0)}})=d_1(C)=2^{2m-1}-2^{m-1}].$$
The desired result follows from
\begin{equation*}
 \sum_{i=0}^{b-1}\left\lceil\frac{2^{2m-1}-2^{m-1}}{2^i}\right\rceil =(2^m-1)\sum_{i=0}^{b-1}\frac{2^{m-1}}{2^i}= (2^b-1)(2^{2m-b}-2^{m-b}).
\end{equation*}
This completes the proof.
\end{proof}
\begin{example}
(i) Let $C$ be the code in Example \ref{ex13} (i). It is easy to verify that $$\ccc=(011001110010000)\in C,$$
the $b$-symbol weight of the codeword $\ccc$ is $9$ and $\overline{\mathcal{I}_b(\ccc)}=\{4,9,12,13,14,15\}$. The shortened code $C_{\overline{\mathcal{I}_b(\ccc)}}$ has parameters $$\left[n=d_2(C)=9,k=b=2,d_1
\left(C_{\overline{\mathcal{I}_b(\ccc)}}\right)=d_1(C)=6\right],$$ which is a Griesmer code.

(ii) When $m=4$ and $b=3$, the Kasami code has parameters $[n=255,k=12,d_1(C)=120]$. Its $3$-symbol weight distribution is
$1+255T^{210}+510T^{214}+510T^{218}+765T^{222}+255T^{224}+765T^{226}
+510T^{230}+510T^{234}+15T^{238},$
and the minimum $3$-symbol weight is $d_3(C)=210.$ The $3$-symbol weight of the codeword
$\ccc=
(1 1 1 1 1 0 0 1 1 1 1 0 0 0 0 0 0 1 1 1 0 1 1 0 0 0 1\\ 0 0 0 1 1 1 0 1 0 1 1 0 0
    0 0 0 0 0 0 0 0 0 1 0 0 1 0 1 0 1 1 1 0 0 0 0 0 1 1 1 1 0 0 0 1 0 1 0 1 1 0
    1 0 0 1 1 1 0 0 1 0 1 1 1 0 0 0 1 0 0 0 1 1 0 0 1 1\\ 1 0 0 1 1 1 0 0 0 0 0 1
    1 0 1 0 1 0 0 0 0 1 0 1 0 1 0 0 0 0 0 1 0 1 0 0 1 0 0 1 1 1 1 0 0 0 0 1 0 1
    0 1 1 1 0 1 1 0 1 0 0 0 0 0 0 1 1 0 1 0 0 1 0 1 1 0 1\\ 0 0 1 1 0 0 0 0 0 0 1
    0 1 1 0 1 0 0 0 0 0 1 1 1 1 1 1 0 0 1 0 1 0 1 0 0 0 0 1 1 0 0 0 1 1 1 1 0 1
    1 1 0 0 1 1 0 1 1 0 1 0 1 1 1 0 1 1 1 0 0 1 1 1 1)$ is\\ 210,
and $\overline{\mathcal{I}_b(\ccc)}=\{
12, 13, 14, 15, 24, 28, 39, 40, 41, 42, 43, 44, 45, 46, 47, 60, 61, 62, 69,
92, 96,\\ 111, 112, 113, 122, 123, 131, 132, 133, 148, 149, 164, 165, 166, 167,
186, 187, 188, 189, 198, 199,\\ 200, 216, 217, 222\}$. The shortened code $C_{\overline{\mathcal{I}_b(\ccc)}}$ has parameters $$\left[n=d_3(C)=210,k=b=3,d_1
\left(C_{\overline{\mathcal{I}_b(\ccc)}}\right)=d_1(C)=120\right],$$ which is a Griesmer code.

These results are verified by Magma programs.
\end{example}
\subsection{Case $m<b\leq 2m$}
In this subsection, we have to get rid of these vectors where the denominator $\Omega_j=0$ since $m<b\leq 2m$.
For $1\leq j\leq b-1$, let $\Gamma_1(j)$ be defined by
\begin{equation*}
  \Gamma_1(j)=\left\{(u_1,u_2,\ldots,u_{j-1})\in\F_2^{j-1}
  \left|\Omega_j=0\right.\right\}.
\end{equation*}
\begin{lemma}\label{lem10}
Let $\#\Gamma_1(j)$ denote the size of $\Gamma_1(j)$. Then
\begin{equation*}
  \#\Gamma_1(j)=\left\{
                  \begin{array}{ll}
                    0, & \hbox{if $1\leq j\leq m-1$;} \\
                    1, & \hbox{if $j=m$;} \\
                    2^{j-m-1}, & \hbox{if $m+1\leq j\leq\ b-1$.}
                  \end{array}
                \right.
\end{equation*}
\end{lemma}
\begin{proof}
The first two cases are trivial and omitted. Assume that $m+1\leq j\leq\ b-1$.
Let $f(x)=1+a_1x+\cdots+a_{m-1}x^{m-1}+x^m$ be the minimal polynomial of $\eta$ over $\F_2$. Let $g(x)$ be the polynomial of the form
$g(x)=1+c_1x+\cdots+c_{j-m-1}x^{j-m-1}+x^{j-m}$, where $c_1,\ldots,c_{j-m-1}\in\F_2$. Let $\{1,e_1,e_2,\ldots,e_{j-1},1\}$ be the coefficients of the polynomial $f(x)\cdot g(x)$. It is easy to check that the vector $(e_1,e_2,\ldots,e_{j-1})$ belongs to $\Gamma_1(j)$. The size of $\Gamma_1(j)$ follows since there are $2^{j-m-1}$ choices for $g(x).$
\end{proof}
For any $1\leq j\leq b-1$ and $(\aaa,\bbb)\in\F_{q^2}^*\times\F_q^*$, let $\mathcal{T}_2(j;\aaa,\bbb)$ be defined by
\begin{equation*}
  \mathcal{T}_2(j;\aaa,\bbb)=\left\{
  (u_1,\ldots,
u_{j-1})\in\F_2^{j-1}\setminus\Gamma_1(j)\left|\T_m\left(
\frac{\aaa^{q+1}}{\bbb}\cdot\frac{\Theta_j^{q+1}}{\Omega_j}\right)=1\right.\right\}.
\end{equation*}

The following lemma is intended to simplify the computation in the subsequent proof of Theorem \ref{thm13}.
\begin{lemma}\label{lem15}
Let $\Delta_1$, $\Delta_2$, $\Delta_3$, $\Delta_4$ and $\Delta_5$ be defined by
\begin{eqnarray*}
  \Delta_1 &=& \sum_{j=1}^{m-1}(m+1-j)\sum_{(u_1,\ldots,u_{j-1})\in\F_2^{j-1}}S(\aaa\Theta_j,
\bbb\Omega_j),\\
\Delta_2 &=& \sum_{(u_1,\ldots,u_{m-1})\in\F_2^{m-1}}S(\aaa\Theta_m,
\bbb\Omega_m),\\
\Delta_3 &=& \sum_{j=1}^{m-1}(b-j)\sum_{(u_1,\ldots,u_{j-1})\in\F_2^{j-1}}S(\aaa\Theta_j,
\bbb\Omega_j), \\
  \Delta_4 &=& (b-m)\sum_{(u_1,\ldots,u_{m-1})\in\F_2^{m-1}}S(\aaa\Theta_m,
\bbb\Omega_m), \\
  \Delta_5 &=& \sum_{j=m+1}^{b-1}(b-j)\sum_{(u_1,\ldots,u_{j-1})\in\F_2^{j-1}}S(\aaa\Theta_j,
\bbb\Omega_j),
\end{eqnarray*}
respectively.
Then
\begin{eqnarray*}
  \Delta_1 &=& -3\cdot2^{2m-1}+(2m+1)2^{m-1}+m+2+2^{m+1}\sum_{j=1}^{m-1}(m+1-j)\#\mathcal{T}_2(j;\aaa,\bbb), \\
  \Delta_2 &=& -2^{2m-1}+2^{m-1}+2^{m+1}\#\mathcal{T}_2(m;\aaa,\bbb),\\
  \Delta_3 &=& -(b-m+2)2^{2m-1}+(m+b)2^{m-1}+b+1+2^{m+1}\sum_{j=1}^{m-1}(b-j)\#\mathcal{T}_2(j;\aaa,\bbb),\\
  \Delta_4 &=&
-(b-m)2^{2m-1}-(b-m)2^{m-1}
+(b-m)2^{m+1}\#\mathcal{T}_2(m;\aaa,\bbb),\\
  \Delta_5 &=& -2^{b+m}+(b-m+1)2^{2m}+2^{m+1}\sum_{j=m+1}^{b-1}(b-j)
\#\mathcal{T}_2(j;\aaa,\bbb).
\end{eqnarray*}
\end{lemma}
\begin{proof}By the definition of $\mathcal{T}_2(j;\aaa,\bbb)$ and Lemma \ref{lem11}, we have
\begin{eqnarray*}
  \Delta_1 &=& \sum_{j=1}^{m-1}(m+1-j)(\#\mathcal{T}_2(j;\aaa,\bbb)(q-1)+(2^{j-1}-\#\mathcal{T}_2(j;\aaa,\bbb))(-q-1)) \\
  ~ &=& -(2^m+1)\sum_{j=1}^{m-1}(m+1-j)2^{j-1}+2^{m+1}\sum_{j=1}^{m-1}(m+1-j)\#\mathcal{T}_2(j;\aaa,\bbb) \\
  ~ &=& -3\cdot2^{2m-1}+(2m+1)2^{m-1}+m+2+2^{m+1}\sum_{j=1}^{m-1}(m+1-j)\#\mathcal{T}_2(j;\aaa,\bbb).
\end{eqnarray*}
Using the fact that $\#\Gamma_1(m)=1$, we have
\begin{eqnarray*}
  \Delta_2 &=& \sum_{(u_1,\ldots,u_{m-1})\in\Gamma_1(m)}S(\aaa\Theta_m,0)+
  \sum_{(u_1,\ldots,u_{m-1})\in\F_2^{m-1}
\setminus\Gamma_1(m)}S(\aaa\Theta_m,\bbb\Omega_m)\\
~&=&-1+\#\mathcal{T}_2(m;\aaa,\bbb)(q-1)+(2^{m-1}-1
-\#\mathcal{T}_2(m;\aaa,\bbb))(-q-1)\\
~&=&-2^{2m-1}+2^{m-1}+2^{m+1}\#\mathcal{T}_2(m;\aaa,\bbb).
\end{eqnarray*}
Imitating the computation of $\Delta_1$ and $\Delta_2$, we obtain the values of $\Delta_3$ and $\Delta_4$.
 %
Using the fact that $\#\Gamma_1(j)=2^{j-m-1}$ if $m+1\leq j\leq b-1$, we have
\begin{eqnarray*}
  \Delta_5 &=& \sum_{j=m+1}^{b-1}(b-j)\left(\sum_{(u_1,\ldots,u_{j-1})\in\Gamma_1(j)}S(\aaa\Theta_j,0)\right.+\left.\sum_{(u_1,\ldots,u_{j-1})\in\F_2^{j-1}\setminus\Gamma_1(j)}
S(\aaa\Theta_j,\bbb\Omega_j)\right)\\
~&=&\sum_{j=m+1}^{b-1}(b-j){\big(}2^{j-m-1}(-1)+\#\mathcal{T}_2(j;\aaa,\bbb)(q-1)\\
~&~&+(2^{j-1}-2^{j-m-1}-\#\mathcal{T}_2(j;\aaa,\bbb))(-q-1){\big)}\\
~&=&-2^m\sum_{j=m+1}^{b-1}(b-j)2^{j-1}+2^{m+1}
\sum_{j=m+1}^{b-1}(b-j)\#\mathcal{T}_2(j;\aaa,\bbb)\\
~&=&-2^{b+m}+(b-m+1)2^{2m}+2^{m+1}\sum_{j=m+1}^{b-1}(b-j)\#\mathcal{T}_2(j;\aaa,\bbb).
\end{eqnarray*}
This completes the proof.
\end{proof}
\begin{theorem}\label{thm13}
Let $(\aaa,\bbb)\in\F_{q^2}\times\F_q$ and $m<b\leq 2m.$ Then
\begin{equation*}
  w_b(\ccc(\aaa,\bbb))=\left\{
                         \begin{array}{ll}
                           0, & \hbox{if $\aaa=\bbb=0$;} \\
                           (2^b-1)2^{2m-b} , & \hbox{if $\aaa\neq0$ and $\bbb=0$;} \\
                           2^{2m}-1, & \hbox{if $\aaa=0$ and $\bbb\neq 0$.}
                         \end{array}
                       \right.
\end{equation*}
If $(\aaa,\bbb)\in\F_{q^2}^*\times\F_q^*$, then we have
\begin{eqnarray*}
  w_b(\ccc(\aaa,\bbb)) &=& 2^{2m}+2^m-2^{2m-b}-1
  -\frac{b(S(\aaa,\bbb)+2^m+1)}{2^b}\\
  ~ &~& -\frac{1}{2^{b-m-1}}\sum_{j=1}^{b-1}(b-j)\#\mathcal{T}_2(j;\aaa,\bbb).
\end{eqnarray*}
\end{theorem}
\begin{proof}
The value of $\Theta_j$ can not be zero if $m+1\leq b\leq 2m$. From Eq.(\ref{K3}) and Lemma \ref{lem2}, we have
\begin{equation*}
  w_b(\ccc(\aaa,\bbb))=\left\{
                         \begin{array}{ll}
                           0, & \hbox{if $\aaa=\bbb=0$;} \\
                           \frac{1}{2^b}((2^b-1)(q^2-1)-
(2^b-1)\cdot(-1)), & \hbox{if $\aaa\neq0$ and $\bbb=0$.}
                         \end{array}
                       \right.
\end{equation*}
Assume that $\bbb\neq0$ in the following discussion.
According to Lemma \ref{lem10}, we next need to discuss the case $m<b\leq 2m$. If $b=m+1$, $\aaa=0$ and $\beta\neq0$, we have
\begin{small}
\begin{eqnarray*}
  w_{m+1}(\ccc(0,\bbb)) &=& \frac{1}{2^{m+1}}{\Bigg{[}}(2^{m
  +1}-1)(q^2-1)-bS(0,\bbb) \\
  ~ &~& -\sum_{j=1}^{m}(m+1-j){\Bigg{(}}\sum_{(u_1,\ldots,u_{j-1})\in\Gamma_1(j)}S(0,0)+\sum_{(u_1,\ldots,u_{j-1})\in
  \F_2^{j-1}\setminus\Gamma_1(j)}S(0,\bbb\Omega_j){\Bigg{)}}{\Bigg{]}}\\
  ~ &=& \frac{1}{2^{m+1}}{\Bigg{[}}(2^{m
  +1}-1)(q^2-1)+b(q+1)-\sum_{j=1}^{m-1}(b-j)\sum_{(u_1,\ldots,u_{j-1})
  \in\F_2^{j-1}}S(0,\bbb\Omega_j)\\
  ~ &~& -\sum_{(u_1,\ldots,u_{j-1})
  \in\Gamma_1{m}}S(0,0)-\sum_{(u_1,\ldots,u_{j-1})
  \in\F_2^{m-1}\setminus\Gamma_1{m}}S(0,\bbb\Omega_j){\Bigg{]}}\\
~&=&\frac{1}{2^{m+1}}\left[(2^{m+1}-1)(q^2-1)+b(q+1)+3\cdot2^{2m-1}-(2m+1)
  2^{m-1}-(m+2)\right. \\
  ~&~&\left.-(q^2-1)+(2^{m-1}-1)(q+1)\right]{\hbox{~~~~~~(by Lemma \ref{lem10})}}\\
  ~&=&2^{2m}-1.
\end{eqnarray*}
\end{small}
If $b=m+1$ and $(\aaa,\beta)\in\F_{q^2}^*\times\F_q^*$, we have
\begin{small}
\begin{eqnarray*}
  w_{m+1}(\ccc(\aaa,\bbb)) &=& \frac{1}{2^{m+1}}{\Bigg{[}}(2^{m+1}-1)(2^{2m}-1)-(m+1)S(\aaa,\bbb)
   \\
  ~ &~& -\sum_{j=1}^{m-1}(m+1-j)\sum_{(u_1,\ldots,u_{j-1})\in\F_2^{j-1}}S(\aaa\Theta_j,
\bbb\Omega_j)\\
  ~ &~& -\sum_{(u_1,\ldots,u_{m-1})\in\F_2^{m-1}}S(\aaa\Theta_m,
\bbb\Omega_m){\Bigg{]}}\\
~&=&\frac{1}{2^{m+1}}\left[(2^{m+1}-1)(2^{2m}-1)-(m+1)S(\aaa,\bbb)-
\Delta_1-\Delta_2\right]\\
~&=&2^{2m}+2^{m-1}-1-\frac{(m+1)(S(\aaa,\bbb)+2^m+1)}{2^{m+1}}\\
~&~&-\sum_{i=1}^m(m+1-j)\#\mathcal{T}_2(j;\aaa,\bbb). {\hbox{~~~~~~(by Lemma \ref{lem15})}}
\end{eqnarray*}
\end{small}
In fact, most of the tedious computations in this proof have been given by Lemma\ \ref{lem15}.

If $m+1<b\leq 2m$, $\aaa=0$ and $\bbb\neq0$, we have
\begin{eqnarray*}
  w_{b}(\ccc(0,\bbb)) &=& \frac{1}{2^b}\left[(2^b-1)(q^2-1)+b(q+1) -\sum_{j=1}^{m-1}(b-j)\sum_{(u_1,\ldots,u_{j-1})\in\F_2^{j-1}}
  S(0,\bbb\Omega_j)\right.
\end{eqnarray*}
\begin{eqnarray*}
  ~ &~& -(b-m)\left(S(0,0)+\sum_{(u_1,\ldots,u_{m-1})\in\F_2^{m-1}\setminus\Gamma_1(m)}S(0,\bbb\Omega_m)\right) \\
  ~ &~& \left.-\sum_{j=m+1}^{b-1}(b-j)\left(\sum_{(u_1,\ldots,u_{m-1})\in\Gamma_1(j)}S(0,0)
  +\sum_{(u_1,\ldots,u_{j-1})\in\F_2^{j-1}\setminus\Gamma_1(j)}S(0,\bbb\Omega_j)\right)\right]\\
  ~&=&\frac{1}{2^b}((2^b-1)(2^{2m}-1)+b(q+1)+2^{2m}-b\cdot2^m-b-1){\hbox{~~~~(by Lemma \ref{lem10})}}\\
  ~&=&2^{2m}-1.
\end{eqnarray*}

If $m+1<b\leq 2m$ and $(\aaa,\beta)\in\F_{q^2}^*\times\F_q^*$, combining Theorem \ref{thm4} and the preceding discussion, we have
\begin{eqnarray*}
  w_b(\ccc(\aaa,\bbb)) &=& \frac{1}{2^b}\left[(2^b-1)(q^2-1)-b S(\aaa,\bbb)-\sum_{j=1}^{m-1}(b-j)\sum_{(u_1,\ldots,u_{j-1})\in\F_2^{j-1}}
  S(\aaa\Theta_j,\bbb\Omega_j)\right. \\
  ~ &~& -(b-m)\sum_{(u_1,\ldots,u_{m-1})\in\F_2^{m-1}}
  S(\aaa\Theta_m,\bbb\Omega_m)\\
~&~&\left.-\sum_{j=m+1}^{b-1}(b-j)\sum_{(u_1,\ldots,u_{j-1})\in\F_2^{j-1}}
  S(\aaa\Theta_j,\bbb\Omega_j)\right]\\
~&=&\frac{1}{2^b}\left[(2^b-1)(q^2-1)-b S(\aaa,\bbb)-\Delta_3-\Delta_4-\Delta_5\right]\\
~&=&2^{2m}+2^m-2^{2m-b}-1
  -\frac{b(S(\aaa,\bbb)+2^m+1)}{2^b}\\
~ &~&  -\frac{1}{2^{b-m-1}}\sum_{j=1}^{b-1}(b-j)\#\mathcal{T}_2(j;\aaa,\bbb).
{\hbox{~~~~~~(by Lemma \ref{lem15})}}
\end{eqnarray*}

This completes the proof.
\end{proof}
For $1\leq j\leq b-1$, a trivial upper bound of $\#\mathcal{T}_2(j;\aaa,\bbb)$ is
$$\#\mathcal{T}_2(j;\aaa,\bbb)\leq \left\{
                                     \begin{array}{ll}
                                       2^{j-1}, & \hbox{if $1\leq j\leq m-1$;} \\
                                       2^{m-1}-1, & \hbox{if $j=m$;} \\
                                       2^{j-1}-2^{j-m-1}, & \hbox{if $m+1\leq j\leq b-1$.}
                                     \end{array}
                                   \right.
$$
If there exists $(\aaa_0,\bbb_0)\in\F_{q^2}^*\times\F_q^*$ such that
$S(\aaa_0,\bbb_0)=q-1$ and $\#\mathcal{T}_2(j;\aaa_0,\bbb_0)$ meets the upper bound for any $1\leq j\leq b-1$. Then
\begin{eqnarray*}
  w_b(\ccc(\aaa_0,\bbb_0)) &=& 2^{2m}+2^m-2^{2m-b}-1-\frac{b\cdot2^{m+1}}{2^b}
-\frac{1}{2^{b-m-1}}\sum_{j=1}^{m-1}
(b-j)2^{j-1}\\
~&~&-\frac{(b-m)(2^{m-1}-1)}{2^{b-m-1}}-
\frac{1}{2^{b-m-1}}\sum_{j=m+1}^{b-1}(b-j)(2^{j-1}-2^{j-m-1})
 \\
  ~ &=& 2^{2m}-2^{2m-b}-2^m+1\\
  ~ &=& \dd_b(C).
\end{eqnarray*}
Since $\dd_b(C)$ is a lower bound of $d_b(C)$, the codeword $\ccc(\aaa_0,\bbb_0)$ has the minimum nonzero $b$-symbol weight.
Similar to Corollary \ref{cor13}, we have the following result.
\begin{corollary}
Assume that $m+1\leq b\leq 2m$.
If there exists $(\aaa_0,\bbb_0)\in\F_{q^2}^*\times\F_q^*$ such that
$S(\aaa_0,\bbb_0)=q-1$ and $\#\mathcal{T}_2(j;\aaa_0,\bbb_0)$ meets the upper bound for any $1\leq j\leq b-1$, then the shortened code $C_{\overline{\mathcal{I}_b(\ccc(\aaa_0,\bbb_0))}}$ is a Griesmer code.
\end{corollary}
\begin{proof}
The desired result follows from the parameters of the shortened code, its parameters are
$$\left[n=d_b(C),k=b,
d_1\left(C_{\overline{\mathcal{I}_b\left(\ccc(\aaa_0,\bbb_0)\right)}}
\right)
=d_1(C)\right],$$
where $d_b(C)=\dd_b(C)=2^{2m}-2^m-2^{2m-b}+1$ and $d_1(C)=2^{2m-1}-2^{m-1}.$
\end{proof}

For any $1\leq j\leq b-1,$ let $\mathcal{A}_2(j)$ be defined by
$$\mathcal{A}_2(j)=\left\{\left.\frac{\Theta_j^{q+1}}{\Omega_j}\right|
(u_1,\ldots,u_{j-1})\in\F_2^{j-1}\setminus\Gamma_1(j)\right\}.$$
It is easy to see that $\mathcal{A}_2(j)$ is a subset of $\F_q^*.$
If there exists some $j$ such that the size of $\mathcal{A}_2(j)$ is greater than $2^{m-1}$, then there is no $(\aaa,\bbb)\in\F_{q^2}^*\times\F_q^*$ such that $S(\aaa_0,\bbb_0)=q-1$ and $\#\mathcal{T}_2(j;\aaa_0,\bbb_0)$ meets the upper bound.
\begin{example}
(i) If $m=2$, the Kasami code has parameters $[n=15,k=6,d_1(C)=6]$. When $b=3$, its $3$-symbol weight distribution is
$1+15T^{12}+15T^{13}+30T^{14}+3T^{15},$
and the minimum $3$-symbol weight is $d_3(C)=12>11=\dd_3(C).$
When $b=4$, its $4$-symbol weight distribution is
$1+15T^{13}+15T^{14}+33T^{15},$
and the minimum $4$-symbol weight is $d_4(C)=13>12=\dd_3(C).$

(ii) If $m=3$, the Kasami code has parameters $[n=63,k=9,d_1(C)=28]$. When $b=4$, its $4$-symbol weight distribution is
$1+63T^{55}+63T^{56}+63T^{58}+63T^{59}+126T^{60}+63T^{62}+70T^{63},$
and the minimum $4$-symbol weight is $d_4(C)=55>53=\dd_4(C).$

 These results are verified by Magma programs.
\end{example}

\subsection{Case $2m<b\leq 3m$}
Since $\theta$ is a primitive element of $\F_{q^2}$,  $\Theta_j$ can not be zero if $1\leq b\leq 2m$. However, if $2m<b\leq 3m$, we select vectors $(u_1,\ldots,u_{b-1})$ such that $\Theta_j=0$ since the value of $S(0,\bbb\Omega_j)$ is determined to be $-q-1$.
For $1\leq j\leq b-1$, let $\Gamma_2(j)$ be defined by
\begin{equation*}
  \Gamma_2(j)=\left\{\left.(u_1,u_2,\ldots,u_{j-1})\in\F_2^{j-1}\right|
\Theta_j=0\right\}.
\end{equation*}
\begin{lemma}\label{lem13}
The size of $\Gamma_2(j)$ is
\begin{equation*}
 \#\Gamma_2(j)= \left\{
     \begin{array}{ll}
       0, & \hbox{if $1\leq j\leq 2m-1$;} \\
       1, & \hbox{if $j=2m$;} \\
       2^{j-2m-1}, & \hbox{if $2m+1\leq j\leq b-1$.}
     \end{array}
   \right.
\end{equation*}Moreover, $\Gamma_1(j)\cap\Gamma_2(j)=\emptyset$.
\end{lemma}
\begin{proof}
The proof of the size of $\Gamma_2(j)$ is similar to that of Lemma \ref{lem10}, thus it is omitted here.

Assume that there exists a vector $(v_1,\ldots,v_{j-1})\in\Gamma_1(j)\cap\Gamma_2(j)$. Then $\eta$ and $\theta$ are two roots of the corresponding polynomial $f(x)=1+v_1x+\cdots+v_{j-2}x^{j-2}+x^{j-1}$ over $\F_2$. Since $\eta$ is not a conjugate of $\theta$, the size of the union of the conjugate sets of $\theta$ and $\eta$ are $m+2m=3m$, i.e.,
$$\left|\left\{\eta,\eta^2,\ldots,\eta^{2^{m-1}}\right\}
\cup\left\{\theta,\theta^2,\ldots,
\theta^{2^{2m-1}}\right\}\right|=3m.$$
Then $f(x)$ has $3m$ distinct roots, which contradicts the assumption that $b\leq 3m$ and $1\leq j\leq b-1$. Therefore, $\Gamma_1(j)\cap\Gamma_2(j)=\emptyset$.
\end{proof}
For any $1\leq j\leq b-1$ and $(\aaa,\bbb)\in\F_{q^2}^*\times\F_q^*$, let $\mathcal{T}_3(j;\aaa,\bbb)$ be defined by
\begin{equation*}
  \mathcal{T}_3(j;\aaa,\bbb)=\left\{(u_1,\ldots,
u_{j-1})\in\F_2^{j-1}\setminus\{\Gamma_1(j)\cup\Gamma_2(j)\}\left|
\T_m\left(\frac{\aaa^{q+1}}{\bbb}\cdot\frac{\Theta_j^{q+1}}
{\Omega_j}\right)=1\right.\right\}.
\end{equation*}
The following lemma is intended to simplify the computation in subsequent proof of Theorem \ref{thm19}.
\begin{lemma}\label{lem25}
Let $\Delta_i$ with $6\leq i\leq 11$ be defined by
\begin{eqnarray*}
  \Delta_6 &=& \sum_{j=1}^{m-1}(2m+1-j)\sum_{(u_1,\ldots,u_{j-1})\in\F_2^{j-1}}S(\aaa\Theta_j,\bbb\Omega_j), \\
  \Delta_7 &=& (m+1)\sum_{(u_1,\ldots,u_{m-1})\in\F_2^{m-1}}S(\aaa\Theta_m,\bbb\Omega_m), \\ \Delta_8 &=& \sum_{j=m+1}^{2m-1}(2m+1-j)\sum_{(u_1,\ldots,u_{j-1})
  \in\F_2^{j-1}}S(\aaa\Theta_j,\bbb\Omega_j),\\
  \Delta_9 &=& \sum_{(u_1,\ldots,u_{2m-1})\in\F_2^{2m-1}}
S(\aaa\Theta_{2m},\bbb\Omega_{2m}),\\
  \Delta_{10} &=& \sum_{j=1}^{m-1}(b-j)\sum_{(u_1,\ldots,u_{j-1})\in\F_2^{j-1}}S(\aaa\Theta_j,\bbb\Omega_j), \\
  \Delta_{11} &=& (b-m)\sum_{(u_1,\ldots,u_{m-1})\in\F_2^{m-1}}S(\aaa\Theta_m,\bbb\Omega_m),\\
  \Delta_{12} &=& \sum_{j=m+1}^{2m-1}(b-j)\sum_{(u_1,\ldots,u_{j-1})\in\F_2^{j-1}}S(\aaa\Theta_j,\bbb\Omega_j),\\
  \Delta_{13} &=& (b-2m)\sum_{(u_1,\ldots,u_{2m-1})\in\F_2^{2m-1}}S(\aaa\Theta_{2m},\bbb\Omega_{2m}),\\
  \Delta_{14} &=& \sum_{j=2m+1}^{b-1}(b-j)\sum_{(u_1,\ldots,u_{j-1})\in\F_2^{j-1}}S(\aaa\Theta_j,\bbb\Omega_j),
\end{eqnarray*}
respectively.
Then
\begin{eqnarray*}
  \Delta_6 &=& -(m+3)\cdot2^{2m-1}+(3m+1)2^{m-1}+2m+2+2^{m+1}\sum_{j=1}^{m-1}
(2m+1-j)\#\mathcal{T}_3(j;\aaa,\bbb), \\
  \Delta_7 &=& -(m+1)(2^{2m-1}-2^{m-1}+1)+2^{m+1}(m+1)\#\mathcal{T}_3(m;\aaa,\bbb), \\
  \Delta_8 &=& -3\cdot2^{3m-1}+(m+2)\cdot2^{2m}+2^{m+1}\sum_{j=m+1}^{2m-1}
  (2m+1-j)\#\mathcal{T}_3(j;\aaa,\bbb),\\
\Delta_9 &=&-2^{3m-1}+2^{m+1}\#\mathcal{T}_3(2m;\aaa,\bbb),\\
  \Delta_{10} &=& -(b-m+2)2^{2m-1}+(b+m)2^{m-1}+b+1+2^{m+1}
  \sum_{j=1}^{m-1}(b-j)\#\mathcal{T}_3(j;\aaa,\bbb), \\
  \Delta_{11} &=& -(b-m)2^{2m-1}+(b-m)2^{m-1}+2^{m+1}(b-m)
  \#\mathcal{T}_3(m;\aaa,\bbb),\\
   \Delta_{12} &=& -(b-2m+2)2^{3m-1}+(b-m+1)2^{2m}+2^{m+1}\sum_{j=m+1}^{2m-1}(b-j)
  \#\mathcal{T}_3(j;\aaa,\bbb),
\end{eqnarray*}
\begin{eqnarray*}
 \Delta_{13} &=& -(b-2m)2^{3m-1}+(b-2m)2^{m+1}\#\mathcal{T}_3(2m;\aaa,\bbb),\\
  \Delta_{14} &=& -2^{b+m}+(b-2m+1)2^{3m}+2^{m+1}
  \sum_{j=2m+1}^{b-1}(b-j)\#\mathcal{T}_3(j;\aaa,\bbb).
\end{eqnarray*}
\end{lemma}
\begin{proof}
Similar to the computation of $\Delta_1$, $\Delta_2$ and $\Delta_5$, we obtain the values of $\Delta_6$, $\Delta_7$, $\Delta_8$, $\Delta_{12}$ and $\Delta_{13}$,
By the fact that $\Gamma_1(j)\cap\Gamma_2(j)=\emptyset$, we have
\begin{eqnarray*}
  \Delta_9 &=& \sum_{(u_1,\ldots,u_{2m-1})\in\Gamma_1(2m)}S(\aaa\Theta_{2m},0)
   +\sum_{(u_1,\ldots,u_{2m-1})\in\Gamma_2(2m)}S(0,\bbb\Omega_{2m})\\
  ~&~&+\sum_{(u_1,\ldots,u_{2m-1})\in\F_2^{2m-1}\setminus\{\Gamma_1(2m)\cup\Gamma_2(2m)\}}S(\aaa\Theta_{2m},\bbb\Omega_{2m}) \\
  ~&=&2^{m-1}\cdot(-1)+1\cdot(-q-1)+\#\mathcal{T}_3(2m;\aaa,\bbb)(q-1)\\
  ~&~&+(2^{2m-1}-2^{m-1}-1-\#\mathcal{T}_3(2m;\aaa,\bbb))(-q-1)\\
  ~ &=& -2^{3m-1}+2^{m+1}\#\mathcal{T}_3(2m;\aaa,\bbb).
\end{eqnarray*}
According to Lemma \ref{lem13}, we have
\begin{eqnarray*}
  \Delta_{14} &=& \sum_{j=2m+1}^{b-1}(b-j)\left(\sum_{(u_1,\ldots,u_{j-1})\in\Gamma_1(j)}S(\aaa\Theta_j,0) +\sum_{(u_1,\ldots,u_{j-1})\in\Gamma_2(j)}S(0,\bbb\Omega_j)\right.\\
  ~&~&\left.+\sum_{(u_1,\ldots,u_{j-1})\in\F_2^{j-1}
  \setminus\{\Gamma_1(j)\cup\Gamma_2(j)\}}S(\aaa\Theta_j,\bbb\Omega_j)\right)\\
  ~&=&\sum_{j=2m+1}^{b-1}(b-j)\left(2^{j-m-1}\cdot(-1)
+2^{j-2m-1}(-2^m-1)+\#\mathcal{T}_3(j;\aaa,\bbb)(q-1)\right.\\
~&~&\left.+(2^{j-1}-2^{j-m-1}-2^{j-2m-1}-
  \#\mathcal{T}_3(j;\aaa,\bbb))(-q-1)\right)\\
  ~&=&-2^m\sum_{j=2m+1}^{b-1}(b-j)2^{j-1}+2^{m+1}\sum_{j=2m+1}^{b-1}
(b-j)\#\mathcal{T}_3(j;\aaa,\bbb)\\
  ~&=&-2^{b+m}+(b-2m+1)2^{3m}+2^{m+1}\sum_{j=2m+1}^{b-1}
(b-j)\#\mathcal{T}_3(j;\aaa,\bbb).
\end{eqnarray*}
This completes the proof.
\end{proof}
\begin{theorem}\label{thm19}
Let $(\aaa,\bbb)\in\F_{q^2}\times\F_q$ and $2m< b\leq 3m$.
Then
\begin{equation*}
  w_b(\ccc(\aaa,\bbb))=\left\{
                         \begin{array}{ll}
                           0, & \hbox{if $\aaa=\bbb=0$;} \\
                           2^{2m}-1, & \hbox{if $\aaa\cdot\bbb=0$.}
                         \end{array}
                       \right.
\end{equation*}
If $(\aaa,\bbb)\in\F_{q^2}^*\times\F_q^*$, then
 \begin{eqnarray*}
  w_b(\ccc(\aaa,\bbb)) &=& 2^{2m}+2^m-1-\frac{1}{2^{b-2m}}\\
  ~&~&-\frac{1}{2^b}\left(b(S(\aaa,\bbb)+2^m+1)
  -2^{m+1}\sum_{j=1}^{b-1}(b-j)\#\mathcal{T}_3(j;\aaa,\bbb)\right).
\end{eqnarray*}
\end{theorem}
\begin{proof}
Notice that the size of $\Gamma_1(j)$ is $2^{j-m-1}$ if $2m<b\leq 3m$. By Lemma \ref{lem13}, $|\Gamma_1\cup\Gamma_2|=2^{j-m-1}+2^{j-2m-1}.$
According to Lemma \ref{lem10}, we need to discuss the case when $2m<b\leq 3m$.

If $b=2m+1$, $\aaa\neq0$ and $\bbb=0$, we have
\begin{eqnarray*}
  w_{2m+1}(\ccc(\aaa,0)) &=& \frac{1}{2^{2m+1}}{\Bigg{[}}(2^{2m+1}-1)(q^2-1)-(2m+1)S(\aaa,0)\\
  ~&~&-\sum_{j=1}^{2m-1}(2m+1-j)
  \sum_{(u_1,\ldots,u_{j-1})\in\F_2^{j-1}}S(\aaa\Theta_j,0) \\
  ~ &~&
  -\sum_{(u_1,\ldots,u_{2m-1})\in\Gamma_2(2m)}S(0,0)-\sum_{(u_1,\ldots,u_{2m-1})\in\F_2^{2m-1}\setminus\Gamma_2(2m)}
  S(\aaa\Theta_{2m,0}){\Bigg{]}} \\
  ~ &=& \frac{1}{2^{2m+1}}{\Bigg{[}}(2^{2m+1}-1)(2^{2m}-1)+2m+1+
  \sum_{j=1}^{2m-1}(2m+1-j)2^{j-1} \\
  ~&~&-(q^2-1)-(2^{2m-1}-1)(-1){\Bigg{]}}{\hbox{~~~~(by Lemma \ref{lem13})}}\\
  ~&=&2^{2m}-1.
\end{eqnarray*}
If $b=2m+1$, $\aaa=0$ and $\bbb\neq0$, we have
\begin{eqnarray*}
  w_{2m+1}(\ccc(\aaa,0)) &=& \frac{1}{2^{2m+1}}{\Bigg{[}}(2^{2m+1}-1)(q^2-1)-(2m+1)S(0,\bbb)\\
  ~&~&-\sum_{j=1}^{m-1}(2m+1-j)
  \sum_{(u_1,\ldots,u_{j-1})\in\F_2^{j-1}}S(0,\bbb\Omega_j)\\
  ~&~&-(m+1){\Bigg{(}}\sum_{(u_1,\ldots,u_{1m-1})\in\Gamma_1(m)}S(0,0)
  +\sum_{(u_1,\ldots,u_{2m-1})\in\F_2^{m-1}
  \setminus\Gamma_1(m)}S(0,\bbb\Omega_m){\Bigg{)}} \\
  ~&~&{\Bigg{(}}-\sum_{j=m+1}^{2m-1}(2m+1-j){\Bigg{(}}
  \sum_{(u_1,\ldots,u_{j-1})\in\Gamma_1(j)}S(0,0){\Bigg{)}}\\
  ~&~&
  +\sum_{(u_1,\ldots,u_{j-1})\in\F_2^{j-1}\setminus\Gamma_1(j)}
  S(0,\bbb\Omega_j){\Bigg{)}}{\Bigg{]}}\\
  ~&=&2^{2m}-1.
\end{eqnarray*}
If $b=2m+1$ and $(\aaa,\bbb)\in \F_{q^2}^*\times\F_q^*$, then
\begin{eqnarray*}
  w_{2m+1}(\ccc(\aaa,\bbb)) &=& \frac{1}{2^{2m+1}}{\Bigg{[}}(2^{2m+1}-1)(2^{2m}-1)-(2m+1)S(\aaa,\bbb)
\end{eqnarray*}
\begin{eqnarray*}
  ~ &~& -\sum_{j=1}^{m-1}(2m+1-j)\sum_{(u_1,\ldots,u_{j-1})\in\F_2^{j-1}}
  S(\aaa\Theta_j,\bbb\Omega_j)\\
~ &~& -(2m+1-m)\sum_{(u_1,\ldots,u_{m-1})\in\F_2^{m-1}}S(\aaa\Theta_m,\bbb\Omega_m) \\
~ &~& -\sum_{j=m+1}^{2m-1}(2m+1-j)\sum_{(u_1,\ldots,u_{j-1})\in\F_2^{j-1}}S(\aaa\Theta_j,\bbb\Omega_j) \\
~ &~& -(2m+1-2m)\sum_{(u_1,\ldots,u_{2m-1})\in\F_2^{2m-1}}
S(\aaa\Theta_{2m},\bbb\Omega_{2m}){\Bigg{]}}\\
~&=& \frac{1}{2^{2m+1}}\left[(2^{2m+1}-1)(2^{2m}-1)-
(2m+1)S(\aaa,\bbb)-\sum_{i=6}^9\Delta_i\right]\\
~&=& 2^{2m}+2^m-\frac{3}{2}-\frac{(2m+1)(S(\aaa,\bbb)+2^m+1)}{2^{2m+1}} \\
  ~ &~&-\frac{1}{2^m}\sum_{j=1}^{2m}(2m+1-j)\#\mathcal{T}_3(j;\aaa,\bbb).
\end{eqnarray*}
In fact, most of the tedious computations in this proof have been given by Lemma\ \ref{lem25}.

If $2m+1<b\leq 3m$, $\aaa\neq0$ and $\bbb=0$, we have
\begin{small}
\begin{eqnarray*}
  w_{b}(\ccc(\aaa,0)) &=& \frac{1}{2^{b}}\left[(2^{b}-1)(q^2-1)-bS(\aaa,0)-\sum_{j=1}^{2m-1}(b-j)
  \sum_{(u_1,\ldots,u_{j-1})\in\F_2^{j-1}}S(\aaa\Theta_j,0)\right.\\
  ~&~&-(b-2m)\left(\sum_{(u_1,\ldots,u_{2m-1})\in\Gamma_2(2m)}S(0,0)
  +\sum_{(u_1,\ldots,u_{2m-1})\in\F_2^{2m-1}\setminus\Gamma_2(2m)}S(\aaa\Theta_{2m},0)\right) \\
  ~ &~& -\sum_{j=2m+1}^{b-1}(b-j)\left(\sum_{(u_1,\ldots,u_{j-1})\in\Gamma_2(j)}S(0,0)+
  \left.\sum_{(u_1,\ldots,u_{j-1})\in\F_2^{j-1}\setminus\Gamma_2(j)}S(\aaa\Theta_j,0)\right)\right]\\
  ~&=&2^{2m}-1.
\end{eqnarray*}
\end{small}
If $2m+1<b\leq 3m$, $\aaa=0$ and $\bbb\neq0$, we have
\begin{eqnarray*}
  w_{b}(\ccc(\aaa,0)) &=& \frac{1}{2^{b}}\left[(2^{b}-1)(q^2-1)-bS(0,\bbb)-\sum_{j=1}^{m-1}(b-j)
  \sum_{(u_1,\ldots,u_{j-1})\in\F_2^{j-1}}S(0,\bbb\Omega_j)\right.\\
  ~&~&-(b-m)\left(\sum_{(u_1,\ldots,u_{m-1})\in\Gamma_1(m)}S(0,0) +\sum_{(u_1,\ldots,u_{m-1})\in\F_2^{m-1}\setminus\Gamma_1(m)}
  S(\aaa\Theta_{m,0})\right)
\end{eqnarray*}
\begin{eqnarray*}
  ~&~&\left.-\sum_{j=m+1}^{b-1}(b-j)\left(\sum_{(u_1,\ldots,u_{j-1})\in\Gamma_1(j)}S(0,0)+
  \sum_{(u_1,\ldots,u_{j-1})\in\F_2^{j-1}\setminus\Gamma_1(j)S(0,\bbb\Omega_j)}\right)\right]\\
  ~&=&2^{2m}-1.
\end{eqnarray*}
If $2m+1<b\leq 3m$ and $(\aaa,\bbb)\in\F_{q^2}^*\times\F_q^*$, then
\begin{small}
\begin{eqnarray*}
  w_{b}(\ccc(\aaa,\bbb)) &=& \frac{1}{2^{b}}\left[(2^{2m+1}-1)(2^{2m}-1)-bS(\aaa,\bbb) -\sum_{j=1}^{m-1}(b-j)\sum_{(u_1,\ldots,u_{j-1})\in\F_2^{j-1}}
  S(\aaa\Theta_j,\bbb\Omega_j)\right. \\
~ &~& -(b-m)\sum_{(u_1,\ldots,u_{m-1})\in\F_2^{m-1}}S(\aaa\Theta_m,\bbb\Omega_m)\\
~&~& -\sum_{j=m+1}^{2m-1}(b-j)\sum_{(u_1,\ldots,u_{j-1})
\in\F_2^{j-1}}S(\aaa\Theta_j,\bbb\Omega_j)\\
~ &~& -(b-2m)\sum_{(u_1,\ldots,u_{2m-1})\in\F_2^{2m-1}}
S(\aaa\Theta_{2m},\bbb\Omega_{2m})\\
~ &~& \left.-\sum_{j=2m+1}^{b-1}(b-j)
\sum_{(u_1,\ldots,u_{j-1})
\in\F_2^{j-1}}S(\aaa\Theta_j,\bbb\Omega_j)\right]\\
~ &=& \frac{1}{2^{b}}\left[(2^{2m+1}-1)(2^{2m}-1)-bS(\aaa,\bbb)-
\sum_{i=10}^{14}\Delta_{i}\right]\\
~ &=& 2^{2m}+2^m-1-\frac{1}{2^{b-2m}}
-\frac{1}{2^b}\left(b(S(\aaa,\bbb)+2^m+1)
  -2^{m+1}\sum_{j=1}^{b-1}(b-j)\#\mathcal{T}_3(j;\aaa,\bbb)\right).
\end{eqnarray*}
\end{small}
This completes the proof.
\end{proof}
For $1\leq j\leq b-1$, a trivial upper bound of $\#\mathcal{T}_3(j;\aaa,\bbb)$ is
$$\#\mathcal{T}_3(j;\aaa,\bbb)\leq \left\{
                                     \begin{array}{ll}
                                       2^{j-1}, & \hbox{if $1\leq j\leq m-1$;} \\
                                       2^{m-1}-1, & \hbox{if $j=m$;} \\
                                       2^{j-1}-2^{j-m-1}, & \hbox{if $m+1\leq j\leq 2m-1$;}\\
                                       2^{2m-1}-2^{m-1}-1, & \hbox{if $j=2m$;}\\
                                       2^{j-1}-2^{j-m-1}-2^{j-2m-1}, & \hbox{if $2m+1\leq j\leq b-1$.}
                                     \end{array}
                                   \right.
$$
For any $1\leq j\leq b-1,$ let $\mathcal{A}_3(j)$ be defined by
$$\mathcal{A}_3(j)=\left\{\left.\frac{\Theta_j^{q+1}}{\Omega_j}\right|
(u_1,\ldots,u_{j-1})\in\F_2^{j-1}\setminus\{\Gamma_1(j)\cup\Gamma_2(j)\}\right\}.$$
It is easy to see that $\mathcal{A}_2(j)$ is a subset of $\F_q^*.$
After calculating many numerical examples, we found that $\#\mathcal{A}_3(j)$ is always greater than $2^{m-1}$ when $j\geq 2m+1$, thus $\#\mathcal{T}_3(j;\aaa,\bbb)$ may be not meet the upper bound if $j\geq 2m+1$.
A nice lower bound of $d_b(C)$ is $\dd_b(C)$. From the following numerical examples, $d_b(C)=\dd_b(C)$ is possible.
\begin{example}
(i) When $m=2$, the Kasami code has parameters $[n=15,k=6,d_1(C)=6]$. Its $5$-symbol weight distribution is
$1+15T^{14}+48T^{15},$
and the minimum $5$-symbol weight is $d_5(C)=14=\dd_5(C).$
Its $6$-symbol weight distribution is
$1+63T^{15},$
and the minimum $6$-symbol weight is $d_6(C)=15=\dd_6(C).$

(ii) When $m=3$, the Kasami code has parameters $[n=63,k=9,d_1(C)=28]$. Its $7$-symbol weight distribution is
$1+61T^{61}+64T^{62}+386T^{63},$
and the minimum $7$-symbol weight is $d_7(C)=61>60=\dd_5(C).$

These results are verified by Magma programs.
\end{example}
\subsection{Case $3m<b\leq n$}
The $b$-symbol weight enumerator of the Kasami code is determined if $3m<b\leq n$.
\begin{theorem}\label{thm20}
For $3m<b\leq n$, the $b$-symbol weight enumerator of the Kasami code
is
$$A(T)=1+(2^{3m}-1)T^{q^2-1}.$$
\end{theorem}
\begin{proof}
Every nonzero codeword in a cyclic code is generated by a recursion of degree $k$ and thus has at most $k-1$ consecutive zeroes. Since the Kasami code is a cyclic code with dimension $k=3m$, there is no $b$ consecutive zeroes.  Therefore, $w_b(\ccc(\aaa,\bbb))=q^2-1$ for any $(\aaa,\bbb)\in\F_{q^{2}}\times\F_q\setminus\{(0,0)\}.$
\end{proof}
\section{Conclusion}
It is not difficult to verify that for any $1\leq j\leq b-1$,
\begin{equation*}
     \begin{array}{ll}
       \mathcal{T}_3(j;\aaa,\bbb)=\mathcal{T}_2(j;\aaa,\bbb)=
       \mathcal{T}_1(j;\aaa,\bbb), & \hbox{if $1\leq b\leq m$;} \\
       \mathcal{T}_3(j;\aaa,\bbb)=\mathcal{T}_2(j;\aaa,\bbb), & \hbox{if $m<b\leq 2m$.}
     \end{array}
\end{equation*}
Then we only need one invariant $\#\mathcal{T}_3(j;\aaa,\bbb).$
Combining Theorem \ref{thm1}, Theorem \ref{thm0}, Theorem \ref{cor1}, Theorem \ref{thm13}, Theorem \ref{thm19} and Theorem \ref{thm20}, we have the main result of this paper.
\begin{theorem}\label{thm30}
Let $\ccc(\aaa,\bbb)$ be a codeword of the Kasami code and let $d_b(C)$ denote the minimal $b$-symbol distance of the Kasami code, where $(\aaa,
\bbb)\in\F_{q^2}\times\F_q$.
\begin{itemize}
  \item If $\aaa=\bbb=0$, then $w_b(\ccc(\aaa,\bbb))=0$ for $1\leq b\leq n.$
  \item If $\aaa=0$ and $\bbb\neq0$, then $w_b(\ccc(\aaa,\bbb))=\left\{
                                                                  \begin{array}{ll}
                                                                    (2^b-1)(2^{2m-b}+2^{m-b}), & \hbox{if $1\leq b\leq m$;} \\
                                                                    n, & \hbox{if $m+1\leq b\leq n$.}
                                                                  \end{array}
                                                                \right.
  $
  \item If $\aaa\neq0$ and $\bbb=0$, then $w_b(\ccc(\aaa,\bbb))=\left\{
                                                                  \begin{array}{ll}
                                                                    (2^b-1)2^{2m-b}, & \hbox{if $1\leq b\leq 2m$;} \\
                                                                    n, & \hbox{if $2m+1\leq b\leq n$.}
                                                                  \end{array}
                                                                \right.$
  \item If $\aaa\neq 0$ and $\bbb\neq 0$, then
   $$w_b(\ccc(\aaa,\bbb))=
  \left\{
           \begin{array}{ll}
             2^{2m}+2^{m}-2^{2m-b}-2^{m-b}-\frac{b(S(\aaa,\bbb)+2^m+1)}{2^b} & \hbox{~} \\
            -2^{m+1-b}
\sum_{j=1}^{b-1}(b-j)\#\mathcal{T}_3(j;\aaa,\bbb), & \hbox{if $1\leq b\leq m$;} \\
       2^{2m}+2^m-2^{2m-b}-1-\frac{b(S(\aaa,\bbb)+2^m+1)}{2^b}
              & \hbox{} \\
-\frac{1}{2^{b-m-1}}\sum_{j=1}^{b-1}(b-j)\#\mathcal{T}_3(j;\aaa,\bbb), & \hbox{if $m+1\leq b\leq 2m$;} \\
  2^{2m}+2^m-1-\frac{1}{2^{b-2m}}-\frac{b(S(\aaa,\bbb)+2^m+1)}{2^b}
                            & \hbox{} \\
             -\frac{1}{2^{b-m-1}}\sum_{j=1}^{b-1}(b-j)\#\mathcal{T}_3(j;\aaa,\bbb), & \hbox{if $2m+1\leq b\leq 3m$;} \\
          n   , & \hbox{if $3m+1\leq b\leq n$.}
           \end{array}
         \right.
  $$
\end{itemize}
Moreover, the range of the minimum $b$-symbol distance of the Kasami code is as follows.
\begin{eqnarray*}
  (2^b-1)(2^{2m-b}-2^{m-b}) &\leq& d_b(C)~\leq ~(2^{b}-1)2^{2m-b}, ~~~~{\hbox{if $1\leq b\leq m$;}} \\
  2^{2m}-2^m-2^{2m-b}+1 &\leq& d_b(C)~\leq ~(2^{b}-1)2^{2m-b}, ~~~~{\hbox{if $m+1\leq b\leq 2m$;}} \\
  2^{2m}-2^{3m-b} &\leq& d_b(C)~\leq ~n, ~~~~~~~~~~~~~~~~~~~{\hbox{if $2m+1\leq b\leq 3m$;}} \\
  ~ &~& d_b(C)~=~n, ~~~~~~~~~~~~~~~~~~~{\hbox{if $3m+1\leq b\leq n$}}.
\end{eqnarray*}
\end{theorem}
In this paper, the complete symbol-pair ($b=2$) weight distribution of the Kasami codes is determined (Theorem \ref{symbolpair}). The case $b=2$ is of great interest, but so far only a few symbol-pair weight distributions for cyclic codes have been determined. Determining the symbol-pair weight distribution of cyclic codes is a nice subject we consider in the future.

\section*{Acknowledgement}
This research is supported by Natural Science Foundation of China (12071001), Excellent Youth Foundation of Natural Science Foundation of Anhui Province (1808085J20).
The authors would like to thank Prof. T. Helleseth for helpful discussions.

\end{document}